\documentclass[acmsmall,screen,nonacm]{acmart}
\AtBeginDocument{%
  }

\setcopyright{none}
\copyrightyear{2026}
\acmYear{2026}
\acmConference[PLDI '26]{EGRAPHS 2026}{June 15--19,
  2026}{Boulder, CO}

\usepackage{acronym}
\usepackage{comment}
\usepackage{todonotes}
\usepackage[ruled, noend]{algorithm2e}
\usepackage{mathtools}
\usepackage{graphicx}
\usepackage{booktabs}
\usepackage{pgfplots}
\pgfplotsset{compat=newest}
\usepackage{standalone}

\usepackage{mathpartir} 
\usepackage{listings}

\begin{document}

\newacro{KBO}{Knuth-Bendix Order}
\newacro{KBC}{Knuth-Bendix Completion}
\newacro{AC}{associativity and commutativity}
\title{Associativity and Commutativity in Equality Saturation}

\author{Tarik Rosin}
\email{taro00001@uni-saarland.de}
\orcid{0009-0005-8273-7907}
\affiliation{
	\institution{Saarland University}
	\city{Saarbrücken}
	\country{Germany}
}

\author{Marcel Ullrich}
\email{ullrich@cs.uni-saarland.de}
\orcid{0009-0006-0127-9623}
\affiliation{
	\institution{Saarland University}
	\city{Saarbrücken}
	\country{Germany}
}

\author{Sebastian Hack}
\email{hack@uni-saarland.de}
\orcid{0000-0002-3387-2134}
\affiliation{
	\institution{Saarland University}
	\city{Saarbrücken}
	\country{Germany}
}

\begin{abstract}
	Equality saturation is a promising technique for program optimization which sidesteps the phase ordering problem.
	However, current e-graph implementations grow exponentially large, even for simple examples.
	Many practical applications involve associative and commutative (AC) operators.
	We present an extension of relational e-matching that handles AC operators natively by storing terms as multisets.
	Preliminary results show that equality saturation modulo AC uses asymptotically less memory in certain cases.
\end{abstract}




\received{17 April 2026}

\maketitle

\section{Motivation}
\label{sec:motivation}

In this paper we study the effects of \ac{AC} rules on equality saturation. 
Equality saturation is a technique of exploring equivalent terms by rewriting,
while storing deduplicated terms in an e-graph.

When used for program optimization, equality saturation allows for good flexibility.
It does not require an ordering on terms or another kind of canonicalization strategy.
Rules can be bidirectional without the danger of non-productive cyclic rewriting.
This is because the e-graph acts as a compact database of all explored programs.
Equality saturation can be interrupted at any time to extract the best program discovered so far.

However, e-graphs tend to explode in size, even for small programs and few rewrite rules.
The e-graph is explicitly storing the entire search space.
Two prominent rewrite rules, \emph{associativity} and \emph{commutativity}, are known to cause combinatorial blowup.
\begin{align}
	(x + y) + z & = x + (y + z) \tag{A} \\
	x + y       & = y + x       \tag{C}
\end{align}

\ac{AC} is ubiquitous in applications of equality saturation such as
hardware synthesis \citep{10549954,10.1145/3582016.3582032},
program optimization \citep{nandi2021ruler,panchekha2015herbie},
and theorem proving \citep{demoura2007ematchsmt}.

\begin{example}
	Consider only the associative law and $x \cdot 0 = 0$ with the initial term $a \cdot 0$.
	First, equality saturation determines that $a \cdot 0 = 0$.
	In the next iteration, the associative pattern would match on $a \cdot (a \cdot 0)$ by backwards substituting $a \cdot 0$ for $0$ in the initial term.
	This produces the new term $(a \cdot a) \cdot 0$ and the same process repeats indefinitely.
\end{example}

This paper reports on an extension of e-graphs to mitigate the blowup caused by \ac{AC} rules.
We first present bounds on the AC blow-up of standard e-graphs.
We then show how to extend the database view of e-graphs to support efficient e-matching and leverage an existing algorithm for AC congruence closure to support efficient rebuilding.

\subsection{Bounds}

In traditional equality saturation the term $x_1 + x_2 + \ldots + x_n$ produces a saturated e-graph
with exponentially many e-nodes and e-classes.
In this section, we give bounds on the number of e-nodes and e-classes in the presence of \ac{AC} rules.

\begin{theorem}
	For a root term $t_r = x_1 + x_2 + \dots + x_n$, the number of e-classes in an AC-saturated e-graph is given by the function $2^n - 1$.
\end{theorem}
\begin{proof}
	Each equivalence class represents a single variable or a sum of distinct variables, regardless of order and brackets.
	A simple counting argument is that each e-class corresponds to a non-empty subset of variables
	and hence $| \mathcal{P}(\Sigma_n) \setminus \{ \emptyset \} | = 2^n - 1 $.
\end{proof}

\begin{theorem}\label{thm:ac-enodes}
	For a root term $t_r = x_1 + x_2 + \dots + x_n$, the number of e-nodes in an AC-saturated e-graph is given by the function $3^n - 2^{n+1} + n + 1$.
\end{theorem}
\begin{proof}
	Consider an arbitrary non-leaf e-node $t_k = c_1 + c_2$ which sums together $k \geqslant 2$ variables.
	Observe that the chosen variables can be distributed among $c_1, c_2$ in precisely $2^k - 2$ many ways.
	Thus, we can quantify the total number of e-nodes by counting each set of variables multiplied by the number of possible e-nodes.

	$$\sum_{k=2}^{n} \binom{n}{k} (2^k - 2) = 1 + \sum_{k=0}^{n} \binom{n}{k} 2^k 1^{n - k} - 2 \cdot 2^n = 3^n - 2^{n+1} + 1$$

	Adding the leaf e-nodes for the $n$ variables yields $3^n - 2^{n+1} + n + 1$.
\end{proof}

\section{Approach}
\label{sec:algorithm}

Our algorithm builds on previous work by \citet{zhang2022relational}, which introduced relational e-matching.
The modification we propose is adding specific indexed relations for \ac{AC} operators.
Before discussing our approach, we briefly recap the framework of existing relational e-matching approaches.

\subsection{Relational E-Matching}

In relational e-matching, the e-graph is a database with one relation per operator
and the rewrite patterns are translated to conjunctive queries.
We use the \emph{generic join} algorithm due to its simplicity.

\begin{example}
	Take the pattern $x + (-x)$ as an example.
	The pattern is lowered to the conjunctive query $Q(r, x) = R_\texttt{add}(r, x, y) \land R_\texttt{neg}(y, x)$.
	The first argument is always the e-class id, and the remaining arguments define the e-node.
	The variable $r$ is the patterns root, the e-class id of the matched term.
	The variable $y$ is the e-class id of $(-x)$ and was added during lowering.
	For an arbitrary variable ordering, we can compile the query to a concrete algorithm.
	\begin{algorithm}
		\SetKw{Yield}{yield}
		\SetCommentSty{medium}
		\caption{Generic Join for the pattern $x + (-x)$ with variable ordering $x < y < r$}
		$X \gets \pi_2 R_\texttt{add} \cap \pi_2 R_\texttt{neg}$
		\tcp*[f]{Collect all candidates of x by intersecting the two relations} \\
		\For{$x \in X$}{
			$Y \gets \pi_3 \sigma_{2=x} R_\texttt{add} \cap \pi_1 \sigma_{2=x} R_\texttt{neg}$
			\tcp*[f]{Select x and then collect all y's given x} \\
			\For{$y \in Y$}{
				$R \gets \pi_1 \sigma_{3=y} \sigma_{2=x} R_\texttt{add}$
				\tcp*[f]{Collect r given x and y} \\
				\For{$r \in R$}{
					\Yield{$(r, x)$} \tcp*[f]{All candidates survived, so we output this match}
				}
			}
		}
	\end{algorithm}
\end{example}

The generic join algorithm proceeds one variable at a time and determines all values this variable can take.
Next, it branches by fixing the variable to a concrete value and moves on to the next variable.
In order to make this procedure efficient, trie indices are built before the matching phase.

\begin{lstlisting}[basicstyle=\ttfamily\footnotesize]
pub trait TrieCursor {
    // advance the cursor down along the `key` edge
    fn select(&mut self, key: Id);
    fn unselect(&mut self);

    // tests whether the current level has `key`
    fn contains(&self, key: Id);

    // returns an iterator over all keys on the current level
    fn keys(&self) -> Iterator<Id>;
}
\end{lstlisting}

These four functions are sufficient to implement an index that supports generic join.
If all functions run in worst-case logarithmic time, the resulting join algorithm is considered worst-case optimal.
Our novelty lies in having a separate index structure for AC operators.

\subsection{E-Graphs modulo AC}

In equality saturation modulo AC, we do not run the associative and commutative patterns during the e-matching phase.
Instead, these patterns are structural properties of AC relations.

An AC relation is a list of tuples, where tuples are allowed to have variadic arity.
Each tuple stores an e-class id and a multiset of arguments of a term, e.g. $+(\{ x_1, x_2, x_3 \}) \in c$.
When inserting a term into the e-graph, AC operators are flattened into a multiset.
Two terms that are AC-equivalent now share the same multiset e-node.
\newline

We introduce two more concepts: \emph{explicit} and \emph{implicit} representation.
An e-node is explicitly represented if the e-node is stored as a tuple in one of the relations.
In contrast, an e-node may be implicitly represented if it is a subterm of an explicitly represented e-node.
The distinction is useful because it allows us to only store a fraction of the total e-nodes.

We now describe how to implement the trie cursor interface of such relations.
An incomplete but efficient scheme is given by using the multiset operations to implement the interface.
Testing containment is the same as testing whether the multiset contains the key.
Iterating over the keys becomes iterating over the unique elements of the multiset.
Selecting a certain key is removing said key once from the multiset and reinserting it for the unselect operation.
To support this kind of indexing, we need to focus on a single multiset at a time.
Therefore, we add an auxiliary variable $t$ to each AC operator when compiling the pattern.
The variable $t$ selects a multiset and is always first in the variable ordering.

\begin{example}
	The pattern $x + (-x)$ will now be compiled to $Q(r, x) = R_\texttt{add}(r, x, y, t) \land R_\texttt{neg}(y, x)$.
	Here $t$ ranges over all multisets in the relation and we enforce it to be first in the variable ordering.
	\begin{algorithm}
		\SetKw{Yield}{yield}
		\SetCommentSty{medium}
		\caption{Generic Join for the pattern $x + (-x)$ with variable ordering $t < x < y < r$}
		$T \gets \pi_4 R_\texttt{add}$
		\tcp*[f]{t selects a multiset to simplify matching in the inner loops} \\
		\For{$t \in T$}{
			$X \gets \pi_2 \sigma_{4=t} R_\texttt{add} \cap \pi_2 R_\texttt{neg}$
			\tcp*[f]{$R_\texttt{add}$ now has a fixed multiset to search in} \\
			\For{$y \in X$}{
				$Y \gets \pi_3 \sigma_{2=x} \sigma_{4=t} R_\texttt{add} \cap \pi_1 \sigma_{2=x} R_\texttt{neg}$ \\
				\For{$x \in Y$}{
					$R \gets \pi_1 \sigma_{3=y} \sigma_{2=x} \sigma_{4=t} R_\texttt{add}$ \\
					\For{$r \in R$}{
						\Yield{$(r, x)$}
					}
				}
			}
		}
	\end{algorithm}
\end{example}

This approach is not complete because each variable can only be bound to an e-class id but not to a subset of e-class ids.
A counterexample would be trying to match the above pattern on $a + b + -(a + b)$.
When the outer sum is flattened, the pattern variable $x$ may not be bound to $a + b$.
Further research is required to see how significant the loss of completeness is.

\subsection{Rebuilding}

Rebuilding refers to restoring key invariants of the E-Graph.
The most significant such invariant is \emph{congruence closure} $a = b \implies f(a) = f(b)$,
meaning two terms are equal if their constituent terms are equal.

In equality saturation we might discover that $a = b$ via a rewrite rule.
The process of rebuilding then traverses the relations to find parents, which are now also equal.
There is a long history of congruence closure algorithms and in particular AC congruence closure algorithms.
In the literature these are often presented as completion-like procedures that deal with sets of ground equations.
We can understand the e-graph as a set of ground rewrite rules by viewing each e-node as a term that rewrites to its e-class id.
For example, let $f(a, b, c) \in u$ be an e-node contained in the e-class $u$.
In the completion-based presentation, this is indicated by a rule $f(a, b, c) \rightarrow u$
where e-class ids are additional constants distinct from the signature.

In light of this, two terms are proven equivalent by the e-graph if they rewrite to the same canonical constant.
The task of rebuilding is to restore convergence of the set of ground rewrite rules that constitute the e-graph.
This setting extends to AC symbols, namely by rewriting modulo AC.

We use the modular AC congruence closure algorithm by \citet{kapur2021accongruence}.
The algorithm is modular in the sense that each relation is being rebuilt separately.
If one relation finds a new equivalence during rebuilding, it updates the union-find data structure,
causing the other relations to rebuild with the new equivalence.
For a single AC relation, the algorithm works as follows:

\begin{enumerate}
	\item For an AC relation, we mark all e-nodes in the relation as dirty.
	\item Take an arbitrary dirty e-node $f(M) \rightarrow c$ and canonicalize it using the union-find.
	\item Next, simplify $f(M) \rightarrow c$ using any clean e-nodes.
	\item Collect all critical pairs between $f(M) \rightarrow c$ and any clean e-node.
	      Non-trivial critical pairs create new e-nodes which get added to the relation and marked as dirty.
	\item Mark $f(M) \rightarrow c$ as clean and reduce all clean e-nodes with it.
	\item If no more e-nodes are dirty we are done, otherwise repeat from step 2.
\end{enumerate}

Kapur's rebuilding algorithm produces a reduced convergent rewrite system.
The e-graph will contain only minimal e-nodes such that all terms are still represented.

\section{Preliminary Results}
\label{sec:results}

We evaluate how equality saturation modulo AC scales compared to egg \citep{willsey2021egg} as baseline implementation.
We insert a term into an empty e-graph and measure the size after saturating.
The theory is an abelian group extended with a homomorphism $h(ab) = h(a) \cdot h(b)$.
For equality saturation modulo AC we only count the explicitly represented e-nodes,
because implicitly represented e-nodes require no memory.
Figure \ref{fig:egraph-sizes} displays the number of e-nodes and e-classes on small instances.
Table \ref{tab:scaling-laws} summarizes the asymptotic scaling behaviour of equality saturation modulo AC.
Our results demonstrate a significant reduction in the size of an e-graph compared to standard equality saturation.





\begin{figure}[h]
	\centering
	\caption{E-graph sizes for selected terms comparing \textup{egg} and \textup{eqsat-mod-ac}.}
	\begin{tikzpicture}
	\begin{axis}[
		xbar,
		bar width=8pt,
		width=10cm,
		height=14cm,
		xlabel={Count},
		ytick={1,2,4,5,7,8,10,11,13,14},
		yticklabels={e-classes, e-nodes, e-classes, e-nodes, e-classes, e-nodes, e-classes, e-nodes, e-classes, e-nodes},
		yticklabel style={font=\small},
		extra y ticks={1.5, 4.5, 7.5, 10.5, 13.5},
		extra y tick labels={$h(abcd)$, $h(abc)$, $aba^{-1}b^{-1}$, $abcd$, \shortstack[c]{$a \cdot \ldots \cdot a$\\\scriptsize$10\times$}},
		extra y tick style={
				tick style={draw=none},
				ticklabel style={font=\small, anchor=west, xshift=-9em},
			},
		xmin=0,
		enlarge y limits=0.1,
		nodes near coords,
		nodes near coords style={font=\small},
		legend pos=north east,
		grid=major,
		grid style={opacity=.3},
		]

		\addplot[fill=orange!60] coordinates {(30,1) (91,2) (14,4) (24,5) (6,7) (13,8) (5,10) (5,11) (2,13) (2,14)};
		\addplot[fill=blue!40] coordinates {(30,1) (119,2) (14,4) (34,5) (18,7) (70,8) (16,10) (55,11) (11,13) (47,14)};
		\legend{eqsat-mod-ac, egg}

	\end{axis}
\end{tikzpicture}
	\label{fig:egraph-sizes}
\end{figure}

\begin{table}[h]
	\centering
	\caption{Asymptotic e-graph size bounds for selected terms}
	\label{tab:scaling-laws}
	\begin{tabular}{lrrrrl}
		\toprule
		                              & \multicolumn{2}{c}{egg} & \multicolumn{2}{c}{eqsat-mod-ac} &                                            \\
		\cmidrule(lr){2-3} \cmidrule(lr){4-5}
		Term                          & e-nodes                 & e-classes                        & e-nodes  & e-classes & Theory              \\
		\midrule
		$x_1 + x_1 + \ldots + x_1$    & $O(n^2)$                & $O(n)$                           & $O(1)$   & $O(1)$    & AC                  \\
		$x_1 + x_2 + \ldots + x_n$    & $O(2^n)$                & $O(2^n)$                         & $O(n)$   & $O(n)$    & AC                  \\
		$h(x_1 + x_2 + \ldots + x_n)$ & $O(2^n)$                & $O(2^n)$                         & $O(2^n)$ & $O(2^n)$  & AC, homomorphism    \\
		$x \cdot 0$                   & $\infty$                & $\infty$                         & $O(1)$   & $O(1)$    & AC, $x \cdot 0 = 0$ \\
		\bottomrule
	\end{tabular}
\end{table}

\section{Conclusions}
\label{sec:conclusion}

Equality saturation is a powerful method for equational reasoning,
but its practical adoption is limited due to its exponential blow-up.
Our investigation quantifies the effect of AC on classic equality saturation,
formally proving $O(2^n)$ bounds for both e-nodes and e-classes.

By extending relational e-matching with specific indexed relations for AC operators,
we can store terms as multisets and manage their subterms implicitly.
Our approach outperforms traditional e-graphs, using asymptotically less memory.
We observed that the memory savings are theory-dependent.
In the worst case, our approach still requires an exponential amount of memory.

Further research is required to better understand the border at which our approach breaks down
and to find remedies such that e-graphs become scalable.



\bibliographystyle{ACM-Reference-Format}
\bibliography{references}




\end{document}